%% file: gandalf2013.tex
\documentclass[copyright,creativecommons]{eptcs}

\usepackage[english]{babel}
\usepackage[latin1]{inputenc}
\usepackage{url,breakurl}             
\usepackage{xspace}
\usepackage{amsmath,amssymb,amsthm}
\usepackage{paralist}
\usepackage{multicol}
\usepackage{eucal}
\usepackage{tikz}
\usetikzlibrary{shapes.geometric,arrows}
\usepackage[strings]{underscore}
\usepackage{mathtools}
\usepackage[lined,ruled,linesnumbered]{algorithm2e}
\usepackage{subfig}

\input{macro}

\renewcommand{\quote}[1]{\emph{``#1''}}

\providecommand{\event}{GandALF 2013} 

\title{Improving \hyltl model checking of hybrid systems}
\author{Davide Bresolin
\institute{University of Verona (Italy)
\email{davide.bresolin@univr.it}}
}
\def\titlerunning{Improving \hyltl model checking of hybrid systems}
\def\authorrunning{D. Bresolin}

\begin{document}

\maketitle

\begin{abstract}
The problem of model-checking hybrid systems is a long-time challenge in the scientific community. Most of the existing approaches and tools are either limited on the properties that they can verify, or restricted to simplified classes of systems.
To overcome those limitations, a temporal logic called \hyltl has been recently proposed. The model checking problem for this logic has been solved by translating the formula into an equivalent hybrid automaton, that can be analized using existing tools.
The original construction employs a declarative procedure that generates exponentially many states upfront, and can be very inefficient when complex formulas are involved.
In this paper we solve a technical issue in the construction that was not considered in previous works, and propose a new algorithm to translate \hyltl into hybrid automata, that exploits optimized techniques coming from the discrete \ltl community to build smaller automata.
\end{abstract}


\section{Introduction}

\emph{Hybrid systems} are heterogeneous systems characterized by a tight interaction between discrete and continuous components. Typical examples include discrete controllers that operate in a continuous environment, as in the case of manufacturing plants, robotic systems, and cyberphysical embedded systems. 
Because of their heterogeneous nature, hybrid systems cannot be faithfully modeled by discrete only nor by continuous only formalisms. In order to model and specify them in a formal way, the notion of \emph{hybrid automata} has been introduced~\cite{Alur,maler91from}. Intuitively, a hybrid automaton is a ``finite-state automaton'' with continuous variables that evolve according to dynamics characterizing each discrete state (called a {\em location} or {\em mode}). 
Of particular importance in the analysis of hybrid automata is the \emph{model checking problem}, that is, the problem of verifying whether a given hybrid automaton respects some property of interest. Unfortunately, the model checking problem is computationally very difficult. Indeed, even for simple properties and systems, this problem is not decidable~\cite{henzinger98whats}.

For very simple classes of hybrid systems, like timed automata, the model checking problem can be solved exactly~\cite{timed_automata}. Tools like Kronos \cite{kronos97} and \textsc{UPPAAL}~\cite{uppaal} can be used to verify properties of timed automata. For more complex classes of systems, the problem became undecidable, and many different approximation techniques may be used to obtain an answer, at least in some cases.
Tools like PhaVer \cite{Frehse2008} and SpaceEx~\cite{Frehse2011}  can compute approximations of the reachable set of hybrid automata with linear dynamics, and thus can be used to verify safety properties.
Other tools, like HSOLVER~\cite{hsolver}, and Ariadne~\cite{ijrnc2012}, can manage systems with nonlinear dynamics, but are still limited to safety properties.

We are aware of only very few approaches that can specify and verify complex properties of hybrid systems in a systematic way. A first attempt was made in~\cite{Lamport93}, where an extension of the Temporal Logic of Actions called TLA+ is used to specify and implement the well-known gas burner example. Later on, Signal Temporal Logic (STL), an extension of the well-known Metric Interval Logic to hybrid traces, has been introduced to monitor hybrid and continuous systems~\cite{Maler2004}. More recent approaches include the tool KeYmaera~\cite{Platzer2008}, that uses automated theorem proving techniques to verify nonlinear hybrid systems symbolically, and the logic HRELTL~\cite{Cimatti09}, that is supported by an extension of the discrete model checker NuSMV, but it is limited to systems with linear dynamics. 

To overcome the limitations of the current technologies, an automata-theoretic approach for model checking hybrid systems has been recently proposed~\cite{has2013}. The work is based on an extension of the well-known temporal logic \ltl to hybrid traces called \hyltl. The model checking problem for this logic has been solved by translating the formula into an equivalent hybrid automaton, reducing the model checking problem to a reachability problem that can be solved by existing tools. The original construction employs a declarative procedure that generates exponentially many states upfront, and can be very inefficient when complex formulas are involved.

In this paper we solve a technical issue in the construction that was not considered in previous works by identifying the precise fragment of \hyltl that can be translated into hybrid automata, and
we propose a new algorithm to translate formulas into hybrid automata, that exploits optimized techniques coming from the discrete \ltl community to be more efficient than the original declarative approach.

\section{Preliminaries}

Before formally defining hybrid automata and the syntax and semantics of \hyltl we need to introduce some basic terminology.
Throughout the paper we fix the \emph{time axis} to be the set of non-negative real numbers $\bbR^+$. An \emph{interval} $I$ is any convex subset of $\bbR^+$, usually denoted as $[t_1, t_2] = \{t \in \bbR^+ : t_1 \leq t \leq t_2\}$. 
%
We also fix a countable universal set $\cvV$ of \emph{variables}, ranging over the reals. Given a finite set of variables $X \subseteq \cvV$, a \emph{valuation} over $X$ is a function $\bx: X \mapsto \bbR^n$ that associates a value to every variable in $X$. The set $\Val(X)$ is the set of all valuations over $X$. 

\medskip

A notion that will play an important role in the paper is the one of \emph{trajectory}. A trajectory over a set of variables $X$ is a function $\tau: I \mapsto \Val(X)$, where $I$ is a left-closed interval with left endpoint equal to $0$. We assume trajectories to be differentiable almost everywhere on the domain, and  we denote with $\dot{\tau}$ the corresponding (partial) function giving the value of the derivative of $\tau$ for every point in the interior of $I$ where $\tau$ is differentiable (note that $\dot\tau$ might not be differentiable neither continuous). With $\dom(\tau)$ we denote the domain of $\tau$, while with $\tau.\ltime$ (the \emph{limit time} of $\tau$) we define the supremum of $\dom(\tau)$. The \emph{first state} of a trajectory is $\tau.\fstate = \tau(0)$, while, when $\dom(\tau)$ is right-closed, the \emph{last state} of a trajectory is defined as $\tau.\lstate = \tau(\tau.\ltime)$. We denote with $\trajs(X)$ the set of all trajectories over $X$. 
If $[t,t']$ is a subinterval of $\dom(\tau)$, we denote whith $\tau\rest{[t,t']}$ the trajectory $\tau'$ such that $\dom(\tau') = [0,t'-t]$ and $\tau'(t'') = \tau(t''+t)$ for every $t'' \in \dom(\tau')$. 
%
Given two trajectories $\tau_1$ and $\tau_2$ such that $\tau_1.\ltime < + \infty$, their concatenation $\tau_1 \cdot \tau_2$ is the trajectory with domain $[0, \tau_1.\ltime + \tau_2.\ltime]$ such that $\tau_1 \cdot \tau_2(t) = \tau_1(t)$ if $t \in \dom(\tau_1)$,  $\tau_1 \cdot \tau_2(t) = \tau_2(t-\tau_1.\ltime)$ otherwise. 

\medskip

Variables will be used in the paper to build \emph{constraints}: conditions on the value of variables and on their derivative that can define sets of valuations, sets of trajectories, and jump relations. Formally, given a set of variables $X$, and a set of mathematical operators $\op$ (e.g. $+$, $-$, $\cdot$, exponentiation, $\sin$, $\cos$, \dots), we define the corresponded set of \emph{dotted variables} $\dot{X}$ as $\{\dot{x} | x \in X\}$ and the set of \emph{tilde variables} $\pre{X}$ as $\{\pre{x} | x \in X\}$. We use $\op$, $X$, $\dot{X}$ and $\pre{X}$ to define the following two classes of constraints.

\begin{itemize}

	\item \emph{Jump constraints}: expressions built up from variables in $X \cup \pre{X}$, constants from $\bbR$, mathematical operators from $\op$ and the usual equality and inequality relations ($\leq$, $=$, $>$, \dots). Examples of jump constraints are $x = 4 \pre{y} + \pre{z}$, $x^2 \leq \pre{y}$, $\pre{y} > \cos(y)$.

	\item \emph{Flow constraints}: expressions built up from variables in $X \cup \dot{X}$, constants from $\bbR$, mathematical operators from $\op$ and the usual equality and inequality relations ($\leq$, $=$, $>$, \dots). Examples of flow constraints are $\dot{x} = 4 y + z$, $\dot{x} + y \geq 0$, $\sin(x) > \cos(\dot{y})$.
\end{itemize}



\noindent We use jump constraints to give conditions on pairs of valuations $(\pre{\bx}, \bx)$. Given a jump constraint $c$, we say that $(\pre{\bx}, \bx)$ respects $c$, and we denote it with $(\pre{\bx}, \bx) \cmodels c$, when, by replacing every variable $x$ with its value in $\bx$ and every tilde variable $\pre{x}$ with the value of the corresponding normal variable in $\pre{\bx}$ we obtain a solution for $c$.
Flow constraints will be used to give conditions on trajectories. Given a flow constraint $c$, we say that a trajectory $\tau$ respects $c$, and we denote it with $\tau \cmodels c$, if and only if for every time instant $t \in \dom(\tau)$, both the value of the trajectory $\tau(t)$ and the value of its derivative $\dot{\tau}(t)$ respect $c$ (we assume that $\dot{\tau}(t)$ respects $c$ when $\dot\tau$ is not defined on $t$).

\section{\hyltl: syntax and semantics}\label{sec:hyltl}

The logic \hyltl is an extension of the well-known temporal logic \ltl to hybrid systems. Given a \emph{finite} set of actions $A$ and a \emph{finite} set of variables $X$, the language of \hyltl is defined from a set of \emph{flow constraints} \fcs over $X$ by the following grammar:
\begin{equation}\label{eq:grammar}
\varphi ::= f \in \fcs \mid a \in A \mid  \neg \varphi \mid \varphi \land \varphi 
						\mid \varphi \lor \varphi \mid
						\X \varphi \mid \varphi \U \varphi \mid \varphi \R \varphi 
\end{equation} 

\noindent In \hyltl constraints from \fcs and actions from $A$ take the role of propositional letters in standard temporal logics, $\neg$, $\land$ and $\lor$ are the usual boolean connectives, $\X$, $\U$ and $\R$ are hybrid counterpart of the standard \emph{next}, \emph{until} and \emph{release} temporal operators. 

\medskip

The semantics of \hyltl is given in terms of \emph{hybrid traces} mixing continuous trajectories with discrete events. Formally, given a set of actions $A$ and a set of variables $X$, an \emph{hybrid trace over $A$ and $X$} is any infinite sequence $\alpha = \tau_1 a_1 \tau_2 a_2 \tau_3 a_3 \ldots$ such that $\tau_i$ is a trajectory over $X$ and $a_i$ is an action in $A$ for every $i \geq 1$. For every $i > 0$, the truth value of a \hyltl formula $\varphi$ over $\alpha$ at position $i$ is given by the truth relation $\mmodels$, formally defined as follows:

\begin{itemize}
	\item for every $f \in \fcs$, $\alpha, i \mmodels f$ if and only if $\tau_i \cmodels f$;
	\item for every $a \in A$, $\alpha, i \mmodels a$ iff $i > 1$ and $a_{i-1} = a$;
	\item $\alpha,i \mmodels \neg\varphi$ if and only if $\alpha,i\not\mmodels\varphi$;
	\item $\alpha,i \mmodels \varphi \land \psi$ if and only if $\alpha,i \mmodels \varphi$ and $\alpha,i \mmodels \psi$;
	\item $\alpha,i \mmodels \varphi \lor \psi$ if and only if $\alpha,i \mmodels \varphi$ or $\alpha,i \mmodels \psi$;
	\item $\alpha,i \mmodels \X\varphi$ if and only if $\alpha,i+1 \mmodels \varphi$;
	\item $\alpha,i \mmodels \varphi \U \psi$ if and only if there exists $j \geq i$ such that $\alpha,j\mmodels\psi$, and for every $i\leq k < j$, $\alpha,k\mmodels\varphi$;	
	\item $\alpha,i \mmodels \varphi \R \psi$ if and only if for all $j \geq i$, if for every $i \leq k < j$, $\alpha, k \not\mmodels \varphi$ then $\alpha,j \mmodels \psi$.
\end{itemize}

Other temporal operators, such as the ``always'' operator $\G$ and the ``eventually'' operator $\F$ can be defined as usual:
\begin{eqnarray*}
	\F \varphi = \top \U \varphi & \qquad\qquad & \G \varphi = \neg \F \neg \varphi
\end{eqnarray*}

\subsection{\hyltl with positive constraints}

In this paper we will pay a special attention on formulas of \hyltl where flow constraints from \fcs appears only in positive form, because it will turn out that they constitue the class of formulas that can be translated into hybrid automata. This particular fragment is called \emph{\hyltl with positive flow constraints}, denoted by \hyltlp, and formally defined by the following grammar:
\begin{equation}\label{eq:p_grammar}
\psi ::=  f \in \fcs \mid a \in A \mid  \neg a \in A 
					\mid \psi \land \psi \mid  \psi \lor \psi \mid
						\X \psi \mid \psi \U \psi \mid \psi \R \psi 
\end{equation} 

Despite being a syntactical fragment, \hyltlp turns out to be equally expressive as the full language, at the price of adding an auxiliary action symbol. In the following, given a constraint $c$ we denote with $\bar{c}$ the corresponding ``dual'' constraint obtained by replacing $<$ with $\geq$, $>$ with $\leq$, $=$ with $\neq$, and so on. Notice that a trajectory $\tau$ that satisfies the negation of a flow constraint $\neg c$ does not necessarily satisfy $\bar{c}$. Indeed, by the semantics of \hyltl we have that $\tau \cmodels \neg c$ if \emph{there exists} a time instant $t$ such that $\tau(t) \not\cmodels c$, while $\tau \cmodels \bar{c}$ if \emph{for all} time instants $t$ we have that $\tau(t) \not\cmodels c$. 

\begin{table}[tbp]
\vspace{-\baselineskip}
$$
\setlength{\arraycolsep}{2pt}
\renewcommand{\arraystretch}{1.1}
\begin{array}{rlrlrl}
\hline
	\ \ \text{when } a \in  & A: & \pi(a) = & a	&	\pi(\neg a) = & \neg a	\\
	\text{when } f \in  & \fcs: & \pi(f) = & f \land \X((T \land f) \U \neg T) &
			\pi(\neg f) = & \bar{f} \lor \X (T \U (T \land \bar{f}))  \\
	& & \pi(\varphi \land \psi) =  & \pi(\varphi) \land \pi(\psi) &
			\pi(\varphi \lor \psi) =  & \pi(\varphi) \lor \pi(\psi) \\
 	& & \pi(\varphi \U \psi) =  & (T \lor \pi(\varphi)) \U (\neg T \land \pi(\psi)) & 
			\quad\pi(\varphi \R \psi) =  & (\neg T \land \pi(\varphi)) \R (T \lor \pi(\psi))\ \ \\
	& & \pi(\X\varphi) =  & \X(T \U (\neg T \land \pi(\varphi))) \\
\hline
\end{array}
$$
\vspace{-1.25\baselineskip}
\caption{The translation function $\pi$ from \hyltl to \hyltlp}\label{tab:pitrans}
\end{table}

Hence, given a trajectory $\tau$ with domain $\dom(\tau) = [0, t_{max}]$ such that $\tau \cmodels \neg c$, it is possible to find a point $t \in [0, t_{max}]$ such that $\tau(t) \not\cmodels c$ and we can split $\tau$ into three sub-trajectories $\tau_b$, $\tau_{\bar c}$, $\tau_e$ such that $\tau_b = \tau\rest{[0, t]}$, $\tau_{\bar c} = \tau\rest{[t,t]}$ and $\tau_e = \tau\rest{[t, t_{max}]}$: it is easy to see that $\tau_{\bar c} \cmodels \bar{c}$. In the following, the auxiliary action symbol $T$ will be used to represent the splitting points of trajectories when translating formulas with negated flow constraints to formulas with positive flow constraints only.

Given   a formula of \hyltl in \emph{in negated normal form} $\varphi$, consider the translation function $\pi$ defined in Table~\ref{tab:pitrans}. 
To compare hybrid traces satisfying the original formula $\varphi$ with the ones satisfying $\pi(\varphi)$ we have to remove the occurrences of $T$ from the latter.
To this end, we define a suitable restriction operator over hybrid traces.

\begin{definition}
Let $A$ a set of action, and $B \subset A$. Given a hybrid trace $\alpha = \tau_1 a_2 \tau_2 a_2 \ldots$ over $A$ we define its \emph{restriction} to $B$ as the hybrid trace $\alpha \rest{B}$ obtained from $\alpha$ by first removing the actions not in $B$ and then concatenating adjacent trajectories.
\end{definition}

The following lemma states that $\pi(\varphi)$ is a formula of \hyltlp equivalent to $\varphi$.

\begin{lemma}\label{lem:pitrans}
For every hybrid trace $\alpha$ over $A$ and $X$ and every \hyltl-formula $\varphi$ we have that $\alpha, 1 \mmodels \varphi$ if and only if there exists a hybrid trace $\beta$ over $A \cup \{T\}$ and $X$ such that $\beta\rest{A} = \alpha$ and $\beta, 1 \mmodels \pi(\varphi)$.
 \end{lemma}
 
\begin{proof}
Let $\alpha = \tau_1 a_1 \tau_2 a_2 \ldots$ be an hybrid trace over $A$ such that $\alpha, 1 \mmodels \varphi$, and let \fcs be the set of flow constraints that appears in $\varphi$. We will build a sequence of hybrid traces $\beta_0, \beta_1, \beta_2, \ldots$ over $A \cup \{T\}$ as follows.
\begin{compactenum}
	\item $\beta_0$ is the empty sequence.
	\item For every $i \geq 1$, consider the $i$-th trajectory $\tau_i$ in $\alpha$, and let $C_i = \{f \in \fcs \mid \tau_i \not\cmodels f\}$. Given an enumeration $f_1, \ldots, f_n$ of $C_i$, we have that it is possible to find a set of time instants $t_1, \ldots, t_n$ such that $\tau_i(t_j) \cmodels \bar{f}$ for every $1 \leq j \leq n$. W.l.o.g., we can assume that $\tau_i.\ftime = t_0 \leq t_1 \leq t_2 \leq \ldots \leq t_n \leq t_{n+1} = \tau_i.\ltime$ and we can define the sequence of trajectories $\mu_1, \mu_2, \ldots, \mu_{2n+1}$ such that
\begin{align}\label{eq:mudef}
	\mu_1 =\; & \tau_i\rest{[t_0,t_1]}, &
	\mu_{2j} = \; & \tau_i\rest{[t_j,t_j]}, &
	\mu_{2j+1} = \; & \tau_i\rest{[t_j,t_{j+1}]} &
	\text{for every } 1 \leq j \leq n
\end{align}

\noindent We define $\beta_i = \beta_{i-1} \mu_1 T \mu_2 T \ldots T \mu_{2n+1} a_{i}$.
\end{compactenum}

\noindent The hybrid trajectory we are looking for is the limit trajectory $\beta = \lim_{i \to \infty} \beta_i$. 

Given an index $i$, we will denote by $\alpha^i$ and $\beta^i$ the suffix of $\alpha$ and of $\beta$ starting at position $i$. We show that $\beta$ respects the following property: 
\quote{for every subformula $\psi$ of $\varphi$ and $i \geq 1$, $\alpha, i \mmodels \psi$ iff $\beta, j \mmodels \pi(\psi)$, where $j$ is the unique index such that $\beta^j\rest{A} = \alpha^i$}. The proof is by induction on $\psi$.

\begin{compactitem}
	\item If $\psi = a$ or $\psi = \neg a$ for some $a \in A$ the property holds trivially.
	\item Suppose $\psi = f$ for some $f \in \fcs$. By the semantics, we have that $\tau_i \cmodels f$. Consider now the sequence $\mu_1 T \mu_2 T \ldots T \mu_{2n+1} a_i$ built in the construction of $\beta_i$, and let $j$ be the index of $\mu_1$ in $\beta$. By~\eqref{eq:mudef} we have that $\mu_h \cmodels f$ for every $1 \leq h \leq 2n+1$. This implies that $\beta,j \mmodels f \land \X ((T\land f)\U\neg T)$.
	\item If $\psi = \neg f$ for some $f \in \fcs$ then we have that $\tau_i \not\cmodels f$. Let $\mu_1 T \mu_2 T \ldots T \mu_{2n+1} a_i$ be the sequence built in the construction of $\beta_i$. Since $f \in C_i$, we have that there exists $t_0 \leq t_k \leq t_{n+1}$ such that $\tau_i(t_k) \cmodels \bar{f}$. By~\eqref{eq:mudef}, this implies that $\mu_k \cmodels \bar{f}$. Let $j$ be the index of $\mu_1$ in $\beta$. Two case may arise: either $\mu_k = \mu_1$ and thus $\beta,j \mmodels \bar{f}$, or $\mu_k \neq \mu_1$ and then $\beta,j \mmodels \X(T\U(T\land\bar{f}))$. In both cases the property is satisfied.
	\item The cases of the boolean operators $\vee$ and $\wedge$ are trivial and can be skipped.
	
	\item Suppose $\psi = \psi_1 \U \psi_2$, and let $i$ be such that $\alpha, i \mmodels \psi_1 \U \psi_2$. By the semantics, we have that there exists $k \geq i$ such that $\alpha, k \mmodels \psi_2$ and, for every $i \leq h < k$, $\alpha, h \mmodels \psi_1$. Now, let $j$ and $l$ be the two indexes such that $\beta^j\rest{A}=\alpha^i$ and $\beta^l\rest{A} = \alpha^k$. By inductive hypothesis we can assume that $\beta,l \mmodels \pi(\psi_2)$, while by the definition of the $\rest{A}$ operator we have that $\beta,l \mmodels a_i \neq T$. Hence, $\beta,l\mmodels \neg T \land \pi(\psi_2)$. Consider now any index $m$ such that $j \leq m < l$. Two cases may arise: either $\beta,m \mmodels T$, or not. In the latter case, we have that it is possible to find an index $i \leq h < k$ such that $\beta^m\rest{A} = \alpha^h$. Since $\alpha,h\mmodels \psi_1$, by inductive hypothesis we have that $\beta,m \mmodels \pi(\psi_1)$. Hence, in both cases $\beta,m\mmodels T \lor \pi(\psi_1)$. This proves that $\beta,j\mmodels(T \lor \pi(\psi_1))\U(\neg T \land \pi(\psi_2)) = \pi(\psi)$.
	
	To prove the converse implication, suppose that $\beta,j \mmodels (T \lor \pi(\psi_1))\U(\neg T \land \pi(\psi_2))$. By the semantics, we have that there exists $l \geq j$ such that $\beta,l \mmodels \neg T \land \pi(\psi_2)$ and, for every $j \leq m < l$, $\beta,m \mmodels T \lor \pi(\psi_1)$. Since $\beta, l \mmodels \neg T$ it is possible to find an index $k$ such that $\beta^l\rest{A} = \alpha^k$. Hence, by inductive hypothesis we have that $\alpha,k \mmodels \psi_2$. Now, let $h$ be such that $i \leq h < k$, and consider the index $m$ such that $\beta^m\rest{A} = \alpha^h$.
By the semantics we have that $\beta,m\mmodels T \lor \pi(\psi_1)$. Since, by definition of the restriction operator, $\beta,m\not\mmodels T$, we have that $\beta,m\mmodels \pi(\psi_1)$ and thus, by inductive hypothesis, that $\alpha,h \mmodels \psi_1$. This proves that $\alpha,i \mmodels \psi_1 \U \psi_2$.
	
	\item The cases of the temporal operators $\X$ and $\R$ can be proved by a similar argument.
	\end{compactitem}

\noindent By the property it is immediate to conclude that, since $\alpha, 1 \mmodels \varphi$ then $\beta, 1 \mmodels \pi(\varphi)$.

To conclude the proof, suppose that there exists a hybrid trace $\beta$ such that $\beta, 1 \mmodels \pi(\varphi)$, and let $\alpha = \beta\rest{A}$. By an induction on the structure of $\varphi$ similar to the one above, we can prove that $\alpha, 1 \mmodels \varphi$.
 \end{proof}


%
%

\section{Hybrid Automata}\label{sec:ha}

An hybrid automaton is a finite state machine enriched with continuous dynamics labelling each discrete state (or \emph{location}), that alternates continuous and discrete evolution. In continuous evolution, the discrete state does not change, while time passes and the evolution of the continuous state variables follows the dynamic law associated to the current location. A discrete evolution step consists of the activation of a \emph{discrete transition} that can change both the current location and the value of the state variables, in accordance with the reset function associated to the transition. 

In this section we recap the definition of Hybrid Automata introduced in~\cite{has2013} to solve the model checking problem for \hyltl. 

\begin{definition}\label{def:ha-syntax}
   A \emph{hybrid automaton} is a tuple $\autH=\langle\Loc,X,$ $A,\Edg,\Dyn,\Rst,\Init\rangle$ such that:
   
   \begin{compactenum}
   	\item $\Loc$ is a finite set of \emph{locations};
   	\item $X$ is a finite set of \emph{variables};
		\item $A$ is a finite set of \emph{actions};
		\item $\Edg \subseteq \Loc \times A \times \Loc$ is a set of \emph{discrete transitions};
		\item $\Dyn$ is a mapping that associates to every location $\ell \in \Loc$ a set of flow constraints $\Dyn(\ell)$ over $X \cup \dot{X}$ describing the \emph{dynamics} of $\ell$;
		\item $\Rst$ is a mapping that associates every discrete transition $(\ell,e,\ell') \in \Edg$ with a set of jump constraints $\Rst(\ell,e,\ell')$ over $\pre{X} \cup X$ describing the guard and reset function of the transition;
		\item $\Init \subseteq \Loc$ is a set of \emph{initial locations}.
\end{compactenum}
\end{definition}

%
%
%
%
%
%
%
%

The \emph{state} of a hybrid automaton $\autH$ is a pair $(\ell,\bx)$, where $\ell\in \Loc$
is a location and $\bx \in \Val(X)$ is a valuation for the continuous variables.
A state $(\ell,\bx)$ is said to be \emph{admissible} if $(\ell,\bx) \cmodels \Dyn(\ell)$.
Transitions can be either \emph{continuous}, capturing the continuous evolution of the state, or 
\emph{discrete}, capturing instantaneous changes of the state.

\begin{definition}\label{def:ha-semantics}
Let $\autH$ be a hybrid automaton.
The \emph{continuous transition relation} $\trans{\tau}$ between admissible states, 
where $\tau$ is a bounded trajectory over $X$, is
defined as follows:
\begin{equation}\label{eq:cont-trans}
(\ell,\bx) \trans{\tau} (\ell,\bx')  \iff 
  	\tau.\fstate=\bx \wedge \tau.\lstate=\bx' \wedge \tau \cmodels \Dyn(\ell).
\end{equation}
%
%
The \emph{discrete transition relation} $\trans{a}$ between admissible states, where
$a \in A$, is defined as follows:
\begin{equation}\label{eq:disc-trans}
(\ell,\bx) \trans{a} (\ell',\bx') \iff
  \bx \cmodels \Dyn(\ell) \wedge \bx' \cmodels \Dyn(\ell') \wedge (\bx,\bx') \cmodels \Rst(\ell,a,\ell').
\end{equation}
\end{definition}

The above definitions allows an infinite sequence of discrete events to occur in a finite amount of time (Zeno behaviors). Such behaviors are physically meaningless, but very difficult to exclude completely from the semantics. In this paper we assume that all hybrid automata under consideration do not generate Zeno runs. This can be achieved, for instance, by adding an extra clock variable that guarantees that the delay between any two discrete actions is bounded from below by some constant.
Moreover, we assume that all hybrid automata are \emph{progressive}, that is, that all runs can be extended to an infinite one: it is not possible to stay forever in a location and never activate a new discrete action. 

We can view progressive, non-Zeno hybrid automata as \emph{generators} of hybrid traces, as formally expressed by the following definition.

\begin{definition}\label{def:ha-trace}
Let $\autH$ be a progressive, non-Zeno hybrid automaton, and let $\alpha = \tau_1 a_1 \tau_2 a_2 \ldots$ be a infinite hybrid trace over $X$ and $A$. We say that $\alpha$ is \emph{generated} by $\autH$ if there exists a corresponding sequence of locations $\ell_1 \ell_2 \ldots$ such that $\ell_1 \in \Init$ and, for every $i \geq 1$:
\begin{inparaenum}[(i)]
\item  $(\ell_i,\tau_i.\fstate) \trans{\tau_i} (\ell_i,\tau_i.\lstate)$, and
\item $(\ell_i,\tau_i.\lstate) \trans{a_i} (\ell_{i+1},\tau_{i+1}.\fstate)$.
\end{inparaenum}
\end{definition}

Our definition of hybrid automata admits composition, 
under the assumpion that all variables and actions are shared between the different automata. The formal definition of the parallel composition operator $\|$ can be found in~\cite{has2013}. In this paper it is sufficient to recall that it respects the usual ``compositionality property'', that is, that the set of hybrid traces generated by a composition of hybrid automata corresponds to the intersection of the hybrid traces generated by the components (up to projection to the correct set of actions and variables).

\section{Model checking \hyltl}

In analogy with the classical automata-theoretic approach, in~\cite{has2013} the model checking problem for \hyltl has been solved by translating the \hyltl formula into an equivalent hybrid automaton, enriched with a suitable \emph{B\"uchi acceptance condition} to identify the traces generated by the automaton that fulfills the semantics of \hyltl.

\begin{definition}\label{def:ha-buechi}
 A \emph{Hybrid Automaton with B\"uchi condition (BHA)} is a tuple $\autH=\langle\Loc, X,A,\Edg,\Dyn,$ $\Rst,\Init,\cvF\rangle$ such that $\langle\Loc,X,A,\Edg,\Dyn,\Rst,\Init\rangle$ is a Hybrid Automaton, and   $\cvF \subseteq {\Loc}$ is a finite set final locations.
\end{definition}

We say that a hybrid trace $\alpha = \tau_1 a_1 \tau_2 a_2 \ldots$ is \emph{accepted} by a BHA $\autH$ if there exists an \emph{infinite} sequence of locations $\ell_1 \ell_2 \ldots$ such that:
\begin{compactenum}[(i)]
\item $\ell_1 \in \Init$;
\item for every $i \geq 1$, $(\ell_i,\tau_i.\fstate) \trans{\tau_i} (\ell_i,\tau_i.\lstate)$; 
\item for every $i \geq 1$, $(\ell_i,\tau_i.\lstate) \trans{a_i} (\ell_{i+1},\tau_{i+1}.\fstate)$;
\item there exists $\ell_f \in \cvF$ that occurs infinitely often in the sequence.
\end{compactenum}

\noindent By the above definition, not all sequences generated by the automaton are accepting: only those that respect the additional accepting condition are considered. 

By the definition of the dynamics, hybrid automata can enforce only \emph{positive constraints} on the continuous flow of the system. Hence, they can only recognize formulas of the positive flow fragment of \hyltl, as summarized by the following theorem. 

\begin{theorem}[\cite{has2013}]\label{th:hyltl-to-bha}
Given a \hyltlp formula $\varphi$, it is possible to build a BHA $\autH_\varphi$ that accepts all and only those hybrid traces that satisfies $\varphi$.
\end{theorem}

Theorem~\ref{th:hyltl-to-bha} and Lemma~\ref{lem:pitrans} can be exploited to solve the model checking problem for full \hyltl as follows. 
Let $\autH_S$ be a hybrid automaton representing the system under verification, and let $\varphi$ be the \hyltl formula representing a property that the system should respect. Consider the formula $\neg\varphi$ and its translation $\overline{\varphi} = \pi(\neg\varphi)$. By Lemma~\ref{lem:pitrans} we have that $\overline{\varphi}$ is a formula of \hyltlp that is equivalent to $\neg\varphi$, and thus we can build a BHA $\autH_{\overline\varphi}$ that is equivalent to \emph{the negation of the property}: it accepts all the hybrid traces that \emph{violates} the property we want to verify. Now, if we compose the automaton for the system with the automaton for $\overline\varphi$ we obtain a BHA $\autH_S \| \autH_{\overline\varphi}$ that accepts only those hybrid traces that are generated by the system and violates the property. This means that $\autH_S$ respects the property $\varphi$  if and only $\autH_S \| \autH_{\overline\varphi}$ \emph{does not accept any hybrid trace}. 

It is worth pointing out that the reachability problem of hybrid automata is undecidable. This means that the model checking of \hyltl is an undecidable problem as well (reachability can be expressed by an eventuality formula). 
However, this does not mean that out logic is completely intractable. A number of different approximation techniques have been developed in the past years to obtain an answer to the reachability problem (at least in some cases), and they can be exploited to solve the model checking problem of \hyltl as well. Indeed, $\autH_S \| \autH_{\overline\varphi}$ accepts an hybrid trace if and only if there exists a loop that includes a final location and that is reachable from the initial states. As shown in~\cite{has2013}, this property can be reduced to a reachability property that can be tested by existing tools for the analysis of hybrid automata. The only thing that one needs to do is to write a procedure implementing the construction of $\autH_{\overline\varphi}$, and then send the results to the reachability analysis tool.

\section{An improved construction algorithm}

The algorithm presented in~\cite{has2013} to build a BHA equivalent to a \hyltlp-formula $\varphi$ is based on a declarative construction. While being simple to understand, it suffers of a major drawback from the efficiency point of view: it generates exponentially many locations upfront, even though many of them may be inconsistent, redundant or unreachable. This implies that the resulting BHA can be very big, even for very simple formulas. In this section we describe an improved construction algorithm, based on the following steps:

\begin{enumerate}[\it A.]
	\item the \hyltlp-formula $\varphi$ is first translated into a suitable formula of discrete \ltl ${\gamma(\varphi)}$;
	\item a discrete B\"uchi automaton $\autA_{\gamma(\varphi)}$, equivalent to ${\gamma(\varphi)}$, is built using one of the many optimized tools available in the literature;
	\item a BHA $\autH_\varphi$, equivalent to $\varphi$, is built from $\autA_{\gamma(\varphi)}$.
\end{enumerate}

 The new algorithm improves the original one by building a smaller BHA, thanks to the use of optimized tools for \ltl in step B.

\subsection{From \hyltl to discrete \ltl}

Let \fcs and $A$ be respectively the set of all flow constraints and discrete actions appearing in $\varphi$. 
For the sake of simplicity, we will assume that $\| A \cup \{T\} \| = 2^n - 1$ for some $n \in \bbN$ (if this is not the case, we can always add some fresh action symbols to $A$ that will not appear in the formula). Under this assumption we can represent action symbols from $A \cup \{T\}$ by means of a set of $n$ propositional letters $B = \{b_0, \ldots, b_{n-1}\}$, where every possible combination of the truth values, but the one where all letters are false, uniquely identify one action symbol. For every $a \in A \cup \{T\}$ let $\bb(a)$ be the corresponding encoding. By definition, we put $\bb(T) = \bigwedge_{i=0}^{n-1} b_i$.

If we consider $\ap = \fcs \cup \{b_0, \ldots, b_{n-1}\}$ as a set of propositional letters for discrete \ltl, we have that we can transform any hybrid trace $\alpha = \tau_1 a_1 \tau_2 a_2 \dots$ into a discrete sequence $\Sigma(\alpha) = \sigma_1 \sigma_2 \sigma_3 \dots$ where every element is a subset of $\ap$ defined as follows: $\sigma_1 = \{f \in \fcs \mid \tau_1 \cmodels f\}$; for every $i > 1$, $\sigma_i = \{f \in \fcs \mid \tau_i \cmodels f\} \cup  \{b_j \in B \mid b_j$ holds true in $\bb(a_{i-1})\}$.

\begin{table}[tbp]
\vspace{-\baselineskip}
$$
\setlength{\arraycolsep}{2pt}
\renewcommand{\arraystretch}{1.1}
\begin{array}{rlrlrl}
\hline
	&  &
	\gamma(\varphi) =  & \bigwedge_{i=0}^{n-1} \neg b_i \land \gamma_0(\varphi) \\
\hline
  \gamma_0(a) = & \bb(a)	  &
	\gamma_0(f) = & f &
	\quad \text{when } a \in & A \text{ or } f \in \fcs \\
	\gamma_0(\neg\varphi) =  & \neg \gamma_0(\varphi) &
 	\gamma_0(\varphi \land \psi) =  & \gamma_0(\varphi) \land \gamma_0(\psi) &
	\gamma_0(\varphi \lor \psi) =  & \gamma_0(\varphi) \lor \gamma_0(\psi) \\
	\quad\gamma_0(\X\varphi) =  & \dX(\gamma_0(\varphi))\quad &
 	\quad\gamma_0(\varphi \U \psi) =  & \gamma_0(\varphi) \dU \gamma_0(\psi)\quad & 
	\quad\gamma_0(\varphi \R \psi) =  & \gamma_0(\varphi) \dR \gamma_0(\psi)\quad \\
\hline
\end{array}
$$
\vspace{-1.25\baselineskip}
\caption{The translation function $\gamma$ from \hyltl to \ltl}\label{tab:gammatrans}
\end{table}

Now, let $\gamma(\varphi)$ be the discrete \ltl formula obtained from $\varphi$ by means of the translation function $\gamma$ defined in Table~\ref{tab:gammatrans}. It is easy to see that $\Sigma(\alpha)$ is a model for $\gamma(\varphi)$, as proved by the following lemma.

\begin{lemma}\label{lem:hyltl-to-ltl}
For every hybrid trace $\alpha$, $\alpha, 1 \mmodels \varphi$ if and only if $\Sigma(\alpha), 1 \mmodels \gamma(\varphi)$.
\end{lemma}

\begin{proof}
Let $\varphi$ be a \hyltl formula, and $\alpha$ a hybrid trace. We prove the lemma by showing that the following stronger claim holds:

\begin{center}
for every $i \geq 1$, $\alpha, i \mmodels \varphi$ if and only if $\Sigma(\alpha), i \mmodels \gamma_0(\varphi)$.
\end{center}

\noindent We reason by induction on the structure of $\varphi$:
\begin{compactitem}
\item if $\varphi = a$, with $a \in A$, then $\gamma_0(a) = \bb(a)$ and the claim follows easily by the definition of $\Sigma(\alpha)$;
\item if $\varphi = f$, with $f \in \fcs$, then $\gamma_0(f) = f$ and the claim follows easily by the definition of $\Sigma(\alpha)$;

\item the boolean cases are trivial and thus skipped;

\item when $\varphi = \X\psi$, we have that $\alpha, i \mmodels \X\psi$ iff $\alpha,i+1\mmodels\psi$. By inductive hypothesis we have that $\Sigma(\alpha),i+1\mmodels\gamma_0(\psi)$, from which we can conclude that $\Sigma(\alpha),i\mmodels\dX\,\gamma_0(\psi)$;

\item suppose $\varphi = \psi_1 \U \psi_2$. By the semantic of \hyltl, we have that $\alpha, i \mmodels \psi_1\U\psi_2$ iff there exists $j \geq i$ such that $\alpha,j\mmodels\psi_2$, and for every $i\leq k < j$, $\alpha,k\mmodels\psi_1$. By inductive hypothesis we have that $\Sigma(\alpha),j\mmodels\gamma_0(\psi_2)$ and that $\Sigma(\alpha),k\mmodels\gamma_0(\psi_1)$ for every $i\leq k < j$. Hence, $\Sigma(\alpha), i \mmodels \gamma_0(\psi_1)\dU\gamma_0(\psi_2)$ and the claim is proved.
\item the case when $\varphi = \psi_1 \R \psi_2$ is analogous.
\end{compactitem}
To conclude the proof it is sufficient to consider that, by definition, $\Sigma(\alpha), 1 \mmodels \bigwedge_{i=0}^{n-1} \neg b_i$. Hence, from the claim it is immediate to conclude that $\Sigma(\alpha), 1 \mmodels \bigwedge_{i=0}^{n-1} \neg b_i \land \gamma_0(\varphi)$ if and only if $\alpha, 1 \mmodels \varphi$.
\end{proof}

When $\varphi$ is a formula of \hyltlp we have that also $\gamma(\varphi)$ is a formula where flow constraints appear only in positive form. Hence, $\gamma(\varphi)$ cannot force the negation of a flow constraint to hold in any of the elements $\sigma_i$ of a discrete sequence, as formally stated by the following lemma.

\newcommand{\Rho}{\mathrm{P}}

\begin{lemma}\label{lem:negation-invariance}
Let $\Sigma = \sigma_1 \sigma_2 \ldots$ and $\Rho = \rho_1 \rho_2 \ldots$ be two discrete sequences such that for every $i \geq 1$, $\sigma_i \cap B = \rho_i \cap B$ (the sequences agrees on the propositional letters in $B$) and $\sigma_i \subseteq \rho_i$ (every flow constraint that is true in $\Sigma$ is true also in $\Rho$). Then, for every \ltl formula   $\gamma$ where flow constraints appear only in positive form and index $j \geq 1$, if $\Sigma, j \mmodels \gamma$ then $\Rho, j \mmodels \gamma$.
\end{lemma}

\begin{proof}
Suppose $\Sigma, j \mmodels \gamma$. We prove the claim by induction on the structure of $\gamma$.
\begin{compactitem}
\item If $\gamma = b_k$ or $\gamma = \neg b_k$, for some $b_k \in B$, we have that the claim follows immediately by the fact that $\sigma_j \cap B = \rho_j \cap B$;
\item If $\gamma = f$ for some $f \in \fcs$, by the semantics of \ltl we have that $f \in \sigma_j$. By hypothesis $\sigma_j \subseteq \rho_j$ and this implies that $\Rho, j \mmodels f$;
\item The remaining cases can be easily proved from the inductive hypothesis and the semantics of \ltl. \qedhere
\end{compactitem}
\end{proof}

\subsection{Building the B\"uchi automaton $\autA_{\gamma(\varphi)}$}

Since the seminal work of Vardi and Wolper~\cite{Vardi1986}, translation of \ltl formulas into equivalent B\"uchi automata plays an important role in many model checking and satisfiability checking algorithms. This led to the development of many translation algorithms exploiting several heuristics and optimization techniques. According to the experiments in~\cite{Rozier2010}, two leading tools are LTL2BA~\cite{Gastin2001} and SPOT~\cite{Duret-Lutz2011}. A new version of the former, called LTL3BA, has been recently introduced~\cite{Babiak2012}. According to the authors, it is faster and it produces smaller automata than LTL2BA, while it produces automata of similar quality with respect to SPOT, being usually faster.

We choose to use LTL3BA as the tool for translating the formula $\gamma(\varphi)$ into the B\"uchi automaton $\autA_{\gamma(\varphi)}$, since it is a state-of-the-art tool that is freely available under an open source license. Nevertheless, the high level \hyltlp translation algorithm is independent from the specific tool used to build $\autA_{\gamma(\varphi)}$, and can be easily adapted to use other tools. 

The output of LTL3BA is a B\"uchi automaton $\autA_{\gamma(\varphi)}$ of the form  $\langle Q, q_0, \delta, F \rangle$, where $Q$ is the set of states, $q_0$ is the unique initial state, $\delta$ is the transition relation and $F$ is the set of final states. 
To merge many transitions into a single one, the transitions are labelled with conjunctions of atomic propositions from \ap: the automaton can fire a transition $(q, \beta, q')$ whenever it reads a symbol $\sigma_j$ of the discrete sequence that satisfies the boolean formula $\beta$. Since $\gamma(\varphi)$ is a formula where flow constraints appear only positively, Lemma~\ref{lem:negation-invariance} guarantees that we can assume, without loss of generality, that in the boolean formulas labeling the transitions of $\autA_{\gamma(\varphi)}$ flow constraints appear only positively. 
The following lemma connects the language of $\autA_{\gamma(\varphi)}$ with the set of hybrid traces satisfying $\varphi$.

\begin{lemma}\label{lem:hyltl-to-ba}
Let $\varphi$ be a \hyltlp formula, and $\alpha$ a hybrid trace. Then $\alpha, 1 \mmodels \varphi$ if and only if $\Sigma(\alpha)$ is accepted by $\autA_{\gamma(\varphi)}$.
\end{lemma}

\begin{algorithm}[t!!!]
 \KwIn{$\autA_{\gamma(\varphi)} = \langle Q, q_0, \delta, F \rangle$}
 \KwOut{$\autH_\varphi = \langle\Loc,X,A\cup\{T\},\Edg,\Dyn,\Rst,$ $\Init,\cvF\rangle$}
 \BlankLine

 $\Loc = \emptyset$, $\Edg = \emptyset$\;
 $\cvL = \emptyset$\;
 \ForEach{transition $(q_0, \beta, q) \in \delta$} {%
		\If{$\beta \to \bigwedge_{i=0}^{n-1}\neg b_i$} {
		  $C = \{f \in \fcs \mid \beta \to f\}$\;
			add $(q,C)$ to $\Loc$\;
			add $(q,C)$ to $\Init$\;
			set $\Dyn(q,C) = C$\;
			add $(q,C)$ to $\cvL$\;
		}
	}
	\While{the queue $\cvL$ is not empty} {%
		extract an element $(q, C)$ from $\cvL$\;
		\ForEach{transition $(q, \beta, q') \in \delta$}{
		  $C' = \{f \in \fcs \mid \beta \to f\}$\;
			\If{$(q',C') \not\in \Loc$}{%
				add $(q',C')$ to $\Loc$\;
				set $\Dyn(q',C') = C'$\;
				add $(q',C')$ to $\cvL$\;
			}
			\ForEach{$a \in A \cup \{T\}$}{
				\If{$\beta \to \bb(a)$}{
					add transition $(q,C,a,q',C')$ to $\Edg$\;
					set $\Rst(q,C,a,q',C') = \top$\;
				}
			}
		}
	}
	$\cvF = \{(q,C) \in \Loc \mid q \in F\}$\;
\caption{how to build the BHA equivalent to $\varphi$}\label{alg:Hphi}
\end{algorithm}

\subsection{From $\autA_{\gamma(\varphi)}$ to $\autH_\varphi$}

By Lemma~\ref{lem:hyltl-to-ba}, we have that the language of $\autA_{\gamma(\varphi)}$ \emph{contains} all the discrete sequences $\Sigma(\alpha)$ such that $\alpha$ satisfies $\varphi$. However, $\autA_{\gamma(\varphi)}$ may accepts also ``spurious'' discrete sequences that do not represent a hybrid trace (for instance, sequences where flow constraints are contradictory). Algorithm~\ref{alg:Hphi} accepts as input the discrete automaton $\autA_{\gamma(\varphi)}$ and build a BHA $\autH_\varphi$ that accepts only the hybrid traces satisfying $\varphi$.

The following theorem proves that the algorithm is correct.

\begin{theorem}\label{teo:Hphi-correctness}
Let $\varphi$ be a formula of \hyltlp, and let  $\autH_\varphi$ be the BHA built by Algorithm~\ref{alg:Hphi}. For every hybrid trace $\alpha$, we have that $\autH_\varphi$ accepts $\alpha$ if and only if $\alpha, 1 \mmodels \varphi$.
\end{theorem}

\begin{proof}
Let $\alpha = \tau_1 a_1 \tau_2 a_2 \dots$ be a hybrid trace such that $\alpha, 1 \mmodels \varphi$. By Lemma~\ref{lem:hyltl-to-ba}, we have that $\autA_{\gamma(\varphi)}$ accepts the discrete sequence $\Sigma(\alpha)$. Let $
q_0 \trans{\beta_1} q_1 \trans{\beta_2} q_2 \trans{\beta_2} \ldots$ be an accepting run of $\autA_{\gamma(\varphi)}$ over $\Sigma(\alpha)$. For every $i \geq 1$, let $C_i = \{f \in \fcs \mid \beta_i \to f\}$, and consider the sequence $(q_1, C_1), (q_2, C_2), (q_3, C_3) \ldots$. By Algorithm~\ref{alg:Hphi} we have that: 
\begin{compactenum}
\item every pair $(q_i, C_i)$ of the sequence is a location of $\autH_\varphi$;
\item $(q_1, C_1) \in \Init$;
\item every set of flow constraints $C_i$ is such that $\Dyn(q_i, C_i) = C_i$;
\item the transition $(q_i, C_i, a_i, q_{i+1}, C_{i+1}) \in \Edg$ with reset condition $\top$.
\end{compactenum}

\noindent By definition of $\Sigma(\alpha)$ we have that $\tau_i \cmodels C_i$, and thus we can conclude that for every $i \geq 1$, both \linebreak $(q_i,C_i,\tau_i.\fstate) \trans{\tau_i} (q_{i}, C_{i}, \tau_i.\lstate)$ and $(q_i,C_i,\tau_i.\lstate) \trans{a_i} (q_{i+1}, C_{i+1}, \tau_{i+1}.\fstate)$ are valid transitions of $\autH_\varphi$. This means that $\alpha$ is generated by $\autH_\varphi$. Since $\Sigma(\alpha)$ is accepted by the discrete automaton $\autA_{\gamma(\varphi)}$ is possible to find a location $(q_f, C_f) \in \cvF$ that occurs infinitely often in the sequence. This proves that $\alpha$ is accepted by $\autH_\varphi$.

To conclude the proof, consider a hybrid trace $\alpha = \tau_1 a_1 \tau_2 a_2 \dots$ that is accepted by $\autH_\varphi$, and let $\Sigma(\alpha) = \sigma_1 \sigma_2 \ldots$ be the corresponding discrete sequence. By the semantics of BHA, it is possible to find an accepting sequence of locations $(q_1, C_1), (q_2, C_2), (q_3, C_3) \ldots$ such that $(q_i,C_i,\tau_i.\fstate) \trans{\tau_i} (q_{i}, C_{i}, \tau_i.\lstate)$ and $(q_i,C_i,\tau_i.\lstate) \trans{a_i} (q_{i+1}, C_{i+1}, \tau_{i+1}.\fstate)$ for every $i \geq 1$. 
By Algorithm~\ref{alg:Hphi} we have that there exists an accepting run $q_0 \trans{\rho_1} q_1 \trans{\rho_2} q_2 \trans{\rho_3} \ldots$ of the discrete automaton $\autA_{\gamma(\varphi)}$ over the discrete sequence $\Rho = \rho_1 \rho_2 \ldots$ where $\rho_i = C_i \cup \{b_j \in B \mid b_j$ holds true in $\bb(a_{i-1})\}$ for every $i \geq 1$.
Since every location $(q_i, C_i)$ is such that $\Dyn(q_i, C_i) = C_i$ we have that for every $f \in C_i$, $\tau_i \cmodels f$ and thus that $\rho_i \subseteq \sigma_i$. From Lemma~\ref{lem:negation-invariance} we can conclude that, since $\autA_{\gamma(\varphi)}$ accepts $\Rho$ then $\autA_{\gamma(\varphi)}$ accepts also $\Sigma(\alpha)$.
By Lemma~\ref{lem:hyltl-to-ba} we can conclude that $\alpha, 1 \mmodels \varphi$.
\end{proof}

\section{The improved algorithm at work}

In~\cite{has2013} feasibility of the automaton-based model checking approach has been tested by verifying the well-known Thermostat example against the \hyltl formula $\varphi_\mathit{hyb} = \neg \F \left( x \geq 21 \land \X \on\right)$ corresponding to the property that \quote{it is not possible that the heater turns on when the temperature is above 21 degrees}. 

To verify the example it is necessary to build the automaton for $\neg\varphi_\mathit{hyb} = \F \left( x \geq 21\land \X\on \right) = \top \U \left( x \geq 21 \land \X \on \right)$. 
The original declarative construction builds a BHA with 18 locations. In this section we will apply the new algorithm to the formula and we will show that the resulting BHA is much smaller that the previous one.
Notice that the formula $\neg\varphi_\mathit{hyb}$ is a formula where flow constraints appears only in positive form. Hence, it is not necessary to apply the translation $\pi$ of Table~\ref{tab:pitrans} to obtain a formula of \hyltlp. 
The first step of the translation algorithm is thus the application of function $\gamma$ (Table~\ref{tab:gammatrans}) to obtain the following formula of discrete \ltl:
\begin{equation*}
\gamma(\neg\varphi_\mathit{hyb}) = \neg b_0 \wedge \neg b_1 \wedge \top \dU \left( x \geq 21 \land \dX (b_0 \land \neg b_1) \right),
\end{equation*}

\noindent where we assume that $\bb(on) = b_0 \land \neg b_1$. By using the tool LTL3BA we obtain the B\"uchi automaton $\autA_{\gamma(\neg\varphi_\mathit{hyb})}$ depicted in Figure~\ref{fig:ba-neg}. Then, by applying Algorithm~\ref{alg:Hphi} we can build the BHA depicted in Figure~\ref{fig:bha-neg}. In both pictures initial states/locations are identified by a bullet-arrow while the final states/locations have a double border. The final BHA obtained by the new construction algorithm is made of only 3 location, with a great improvement over the original declarative construction.

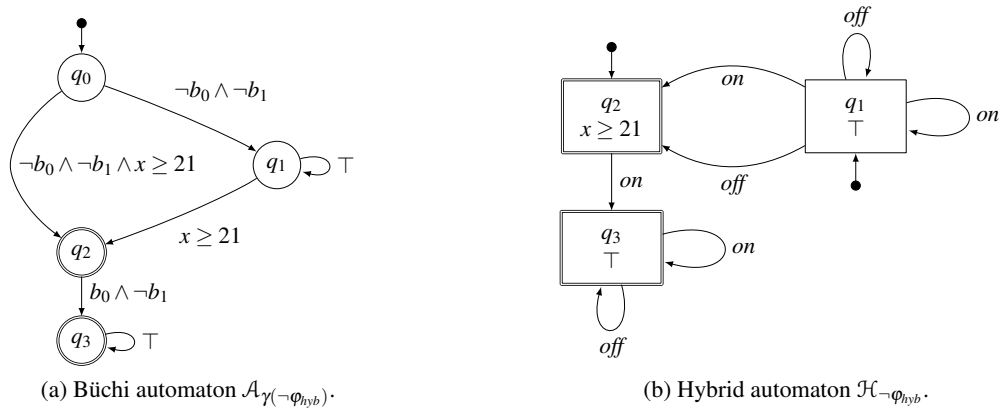
\begin{figure}[tbp]
\ \hfill
\subfloat[B\"uchi automaton $\autA_{\gamma(\neg\varphi_\mathit{hyb})}$.]
	{\scalebox{0.8}{\label{fig:ba-neg}\input{Aphi}}}
\hfill\hfill
\subfloat[Hybrid automaton $\autH_{\neg\varphi_\mathit{hyb}}$.]
	{\scalebox{0.8}{\label{fig:bha-neg}\input{Hphi}}}
\hfill\ 
\caption{The discrete and hybrid automata for $\neg\varphi_\mathit{hyb}$.}
\label{fig:all-neg}
\end{figure}

\medskip

As a second example, consider the globally-eventually formula $\varphi_\mathit{liv} = \G \left(\neg x \geq 18 \to \X\F\on\right)$ expressing the liveness property to \quote{eventually switch the heater $\on$ if the temperature falls below $18$ degrees}. In this case the negation of the property is the formula $\neg\varphi_\mathit{liv} = \F \left(\neg x \geq 18 \land \X\G\neg\on\right) = \top \U  \left(\neg x \geq 18 \land \X(\bot\R\neg\on)\right)$, that do not belongs to the language of \hyltlp. Hence, it is necessary to apply the translation function $\pi$ to obtain the following equivalent formula:
\begin{align*}
\overline{\varphi}_\mathit{liv} & = 
			\pi\Big(\top \U  \big(\neg x \geq 18 \land \X(\bot\R\neg\on)\big)\Big) 
= 	(T \lor \top) \U \Big(\neg T \land \pi\big(\neg x \geq 18 \land \X(\bot\R\neg\on)\big)\Big) \\
& =	\top \U \Big(\neg T \land \pi(\neg x \geq 18) \land \pi\big(\X(\bot\R\neg\on)\big)\Big) \\ 
& = 	\top \U \bigg(\neg T \land \Big(x < 18 \lor \X\big(T \U (T \land x < 18)\big)\Big) \land 
		\X\Big(T\U\big(\neg T \land \pi(\bot\R\neg\on)\big)\Big)\bigg) \\
& = 	\top \U \bigg(\neg T \land \Big(x < 18 \lor \X\big(T \U (T \land x < 18)\big)\Big) \land 
		\X\Big(T\U\big(\neg T \land \bot\R\left(T \lor \neg\on\right)\big)\Big)\bigg)
\end{align*}

\noindent The input formula for LTL3BA is thus
\begin{multline*}
\gamma(\overline{\varphi}_{safe}) = 
	\neg b_0 \land \neg b_1 \land \top \dU \Bigg(\neg (b_0 \land b_1) \land \Big(x < 18 \lor \dX\big(( b_0 \land  b_1) \dU (b_0 \land  b_1 \land x < 18)\big)\Big) 
\\[-10pt]
	\land \dX\bigg((b_0\land b_1) \dU\Big(\neg(b_0\land b_1) \land \bot\dR\big(( b_0 \land  b_1) \lor \neg(b_0 \land \neg b_1)\big)\Big)\bigg)\Bigg)
\end{multline*}

\noindent while the resulting discrete B\"uchi automaton is depicted in Figure~\ref{fig:ba-liv}.
Algorithm~\ref{alg:Hphi} transforms it into the BHA with 5 locations shown in Figure~\ref{fig:bha-liv}.
Notice that, despite the increased complexity of the formula due to the translation into \hyltlp the final result is still of very small size.

\begin{figure}[t!!!]
\centering
\subfloat[B\"uchi automaton $\autA_{\gamma(\overline{\varphi}_\mathit{liv})}$.]
	{\scalebox{0.8}{\label{fig:ba-liv}\input{Aliv}}}\quad
\subfloat[Hybrid automaton $\autH_{\overline{\varphi}_\mathit{liv}}$.]
	{\scalebox{0.8}{\label{fig:bha-liv}\input{Hliv}}}
\caption{The discrete and hybrid automata for $\neg\varphi_\mathit{liv}$.}
\label{fig:all-safe}
\end{figure}
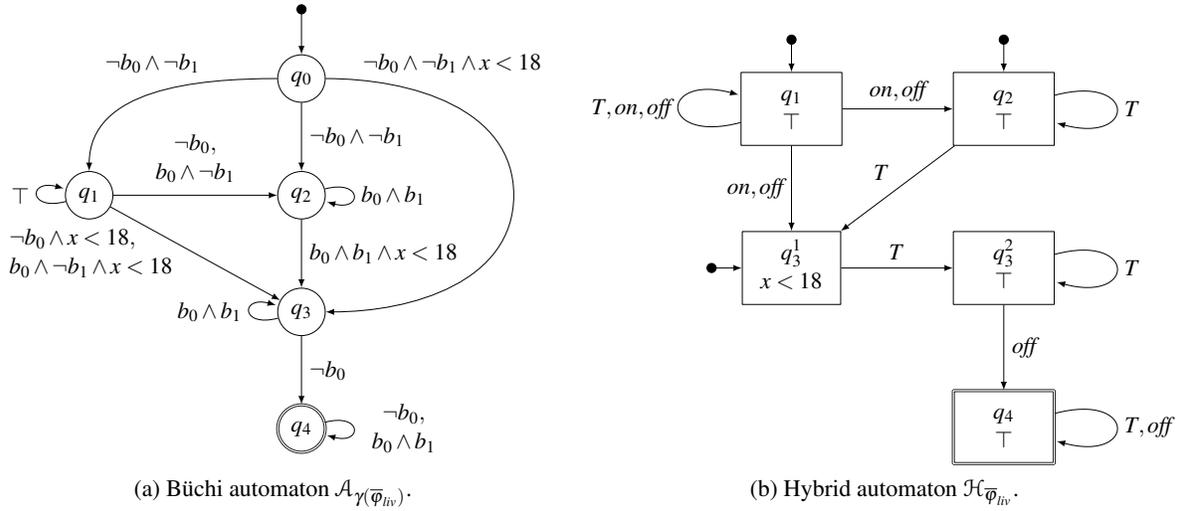

We have verified that the thermostat example given in~\cite{has2013} respects the two example properties  $\varphi_\mathit{hyb}$ and $\varphi_\mathit{liv}$ using the software package PhaVer~\cite{Frehse2008}. Since the system and the automata for the properties are very simple, the computation time was almost instantaneous: less than $0.1 s$ for both formulas on an Intel Core 2 Duo 2.4 GHz iMac with 4 Gb of RAM. 

\section{Conclusion}

In this paper we extended the current research on \hyltl, a logic that is able to express properties of hybrid traces, and that can be used to verify hybrid systems.
We identified the fragment of \hyltl that can be transformed into hybrid automata, that is, the positive flow constraints fragment \hyltlp. Then, we have shown that every property definable in the full language is also definable by \hyltlp. Finally, we developed a new algorithm to translate formulas into hybrid automata, that turned out to be much more efficient than the original declarative algorithm.

This work can be extended in many directions. The expressivity of the logic can be extended by adding jump predicates to the language, to express properties on the reset functions of the system. 
A comprehensive tool support for the logic is currently missing: an implementation of the complete model checking algorithm into the software package Ariadne~\cite{ijrnc2012} is under development.

\bibliographystyle{eptcs}
\bibliography{gandalf2013}
\end{document}

%% file: macro.tex
\newcommand{\hyltl}{\textsf{HyLTL}\xspace}
\newcommand{\hyltlp}{\textsf{HyLTL$^+$}\xspace}
\newcommand{\ltl}{\textsf{LTL}\xspace}
\newcommand{\X}{\ensuremath{\mathbin{\mathbf{X}}}\xspace}
\newcommand{\U}{\ensuremath{\mathbin{\mathbf{U}}}\xspace}
\newcommand{\R}{\ensuremath{\mathbin{\mathbf{R}}}\xspace}
\newcommand{\F}{\ensuremath{\mathbin{\mathbf{F}}}\xspace}
\newcommand{\G}{\ensuremath{\mathbin{\mathbf{G}}}\xspace}

\newcommand{\dX}{\ensuremath{\mathbin{\mathfrak{X}}}\xspace}
\newcommand{\dU}{\ensuremath{\mathbin{\mathfrak{U}}}\xspace}
\newcommand{\dR}{\ensuremath{\mathbin{\mathfrak{R}}}\xspace}

\newcommand{\fcs}{\ensuremath{FC}\xspace}

\newcommand{\mmodels}{\Vdash}
\newcommand{\cmodels}{\vdash}

\newcommand{\ap}{\ensuremath{AP}\xspace}

\newcommand{\on}{\text{\it on}}
\newcommand{\off}{\text{\it off}}

\newcommand{\op}{\ensuremath{OP}\xspace}

\newcommand{\pre}[1]{\tilde{#1}}

\newtheorem{definition}{Definition}
\newtheorem{theorem}{Theorem}

\newtheorem{lemma}{Lemma}

\newcommand{\dims}[1]{k}

\DeclareMathOperator{\dom}{\mathrm{dom}}

\newcommand{\bbN}{\mathbb{N}}

\newcommand{\bbR}{\mathbb{R}}

\newcommand{\cvF}{\mathcal{F}}

\newcommand{\cvL}{\mathcal{L}}

\newcommand{\cvV}{\mathcal{V}}

\newcommand{\autA}{\mathcal{A}}

\newcommand{\autH}{\mathcal{H}}

\newcommand{\Loc}{\mathrm{Loc}}
\newcommand{\Edg}{\mathrm{Edg}}
\DeclareMathOperator{\Dyn}{{Dyn}}

\DeclareMathOperator{\Rst}{{Res}}
\DeclareMathOperator{\Init}{{Init}}

\DeclareMathOperator{\Val}{{Val}}

\newcommand{\fstate}{\mathit{fstate}}
\newcommand{\lstate}{\mathit{lstate}}
\newcommand{\ftime}{\mathit{ftime}}
\newcommand{\ltime}{\mathit{ltime}}

\newcommand{\bb}{\mathbf{b}}

\newcommand{\bx}{\mathbf{x}}

\newcommand{\trans}[1]{\xrightarrow{#1}}

\DeclareMathOperator{\trajs}{\mathit{Trajs}}

\newcommand{\rest}[1]{{\downarrow}_{#1}}

\newcommand{\cut}[1]{ }

%% file: Aphi.tex
\begin{tikzpicture}[>=latex,line join=bevel,scale=0.4,auto]
\node (S3) at (316bp,232bp) [draw,circle] {$q_1$}
 			edge [loop right] node {$\top$} ();
 \node (S2) at (86bp,129bp) [draw,circle,double] {$q_2$};
  \node (S0) at (86bp,25bp) [draw,circle,double] {$q_3$}
  			edge [loop right] node {$\top$} ();
  \node (T0) at (86bp,335bp) [draw,circle] {$q_0$};
	\draw [<-*] (T0) -- ++(0bp,70bp);
  \draw [->] (S3) ..controls (290bp,215bp) and (284bp,211bp)  .. (278bp,208bp) .. controls (224bp,179bp) and (158bp,154bp)  .. node {$x \geq 21$} (S2);
  \draw [->] (S2) ..controls (86bp,91bp) and (86bp,75bp)  .. node {$b_0 \land \neg b_1$} (S0);
  \draw [->] (T0) ..controls (49bp,310bp) and (20bp,285bp)  .. (8bp,256bp) .. controls (0bp,236bp) and (0bp,227bp)  .. node {$\neg b_0 \land \neg b_1 \land x \geq 21$}(8bp,208bp) .. controls (18bp,184bp) and (40bp,162bp)  ..  (S2);
  \draw [->] (T0) ..controls (145bp,314bp) and (219bp,286bp)  .. node {$\neg b_0 \land \neg b_1$} (278bp,256bp) .. controls (281bp,254bp) and (284bp,253bp)  ..  (S3);
\end{tikzpicture}

%% file: Hphi.tex
\begin{tikzpicture}[>=latex,line join=bevel,scale=0.5,auto]
\node (S3) at (316bp,129bp) [draw,rectangle] 
			{$\begin{array}{c} q_1\\ \quad\top\quad\end{array}$}
			edge [loop right] node {$\on$} ()
			edge [loop above] node {$\off$} ();
  \node (S2) at (86bp,129bp) [draw,rectangle,double] 
  			{$\begin{array}{c} q_2 \\ x \geq 21 \end{array}$};
  \node (S0) at (86bp,5bp) [draw,rectangle,double] 
  			{$\begin{array}{c} q_3\\ \quad\top\quad\end{array}$}
			edge [loop right] node {$\on$} ()
			edge [loop below] node {$\off$} ();
  \draw [->] (S3) to[bend right] node {$\on$} (S2) ;
  \draw [->] (S3) to[bend left] node {$\off$} (S2) ;
  \draw [->] (S2) to node {$\on$} (S0);
	\draw [<-*] (S2) -- ++(0bp,70bp);
	\draw [<-*] (S3) -- ++(0bp,-70bp);
%
\end{tikzpicture}

%% file: Aliv.tex
\begin{tikzpicture}[>=latex,line join=bevel,scale=0.5,auto]
  \node (Q0) at (300bp,350bp) [draw,circle] {$q_0$};
		\draw [<-*] (Q0) -- ++(0bp,70bp);
	\node (Q1) at (100bp,240bp) [draw,circle] {$q_1$}
		edge [loop left] node {$\top$} ();
  \node (Q2) at (300bp,240bp) [draw,circle] {$q_2$}
		edge [loop right] node {$b_0 \land b_1$} ();
  \node (Q3) at (300bp,130bp) [draw,circle] {$q_3$}
		edge [loop left] node {$b_0 \land b_1$} (); 
  \node (Q4) at (300bp,20bp) [draw,double,circle] {$q_4$}
		edge [loop right] node {$\begin{array}{c}\neg b_0, \\ b_0 \land b_1\end{array}$} (); 
  \draw [->] (Q0) to[out=0,in=90] node[above=11bp] {$\neg b_0 \land \neg b_1 \land x < 18$} (500bp,240bp) 
  		to[out=-90,in=0] (Q3);
  \draw [->] (Q2) -- node {$b_0 \land b_1 \land x < 18$} (Q3);
  \draw [->] (Q0) -- node {$\neg b_0 \land \neg b_1$} (Q2);
  \draw [->] (Q1) -- node[left=2bp] {$\begin{array}{l}
  		\neg b_0 \land x < 18, \\ b_0 \land \neg b_1 \land x < 18\end{array}$} (Q3);
  \draw [->] (Q3) -- node {$\neg b_0$} (Q4);
  \draw [->] (Q0) to[out=180,in=90] node[above=6bp] {$\neg b_0 \land \neg b_1$} (Q1);
  \draw [->] (Q1) -- node {$\begin{array}{c}\neg b_0,\\ b_0 \land \neg b_1\end{array}$} (Q2);
\end{tikzpicture}

%% file: Hliv.tex
\begin{tikzpicture}[>=latex,line join=bevel,scale=0.5,auto]
	\node (Q1) at (200bp,180bp) [draw,rectangle]
		{$\begin{array}{c} q_1\\ \quad\top\quad\end{array}$}
		edge [loop left] node {$T, \on, \off$} ();
		\draw [<-*] (Q1) -- ++(0bp,70bp);
  \node (Q2) at (400bp,180bp) [draw,rectangle] 
  		{$\begin{array}{c} q_2\\ \quad\top\quad\end{array}$}
		edge [loop right] node {$T$} ();
		\draw [<-*] (Q2) -- ++(0bp,70bp);
  \node (Q3) at (200bp,30bp) [draw,rectangle] 
  		{$\begin{array}{c} q_3^1\\  x < 18\end{array}$};
		\draw [<-*] (Q3) -- ++(-80bp,0bp);
  \node (Q3b) at (400bp,30bp) [draw,rectangle] 
  		{$\begin{array}{c} q_3^2\\ \quad\top\quad\end{array}$}
  		edge [loop right] node {$T$} ();
  \node (Q4) at (400bp,-120bp) [draw,rectangle,double] 
  		{$\begin{array}{c} q_4\\ \quad\top\quad\end{array}$}
  		edge [loop right] node {$T,\off$} ();
  \draw [->] (Q2) -- node[above left] {$T$} (Q3)  ;
  \draw [->] (Q1) -- node[left] {$\on,\off$} (Q3);
  \draw [->] (Q1) -- node {$\on,\off$} (Q2);
  \draw [->] (Q3) -- node {$T$} (Q3b);
  \draw [->] (Q3b) -- node {$\off$} (Q4);
\end{tikzpicture}